\documentclass[twoside,leqno,twocolumn]{article}
\usepackage{ltexpprt}
\usepackage{amsmath}

\providecommand{\qedsymbol}{$\square$}
\newcommand{\mathqed}{\quad\hbox{\qedsymbol}}
\DeclareRobustCommand{\qed}{%
  \ifmmode \mathqed
  \else
    \leavevmode\unskip\penalty9999 \hbox{}\nobreak\hfill
    \quad\hbox{\qedsymbol}%
  \fi
}
\renewenvironment{proof}{\begin{@proof}}{\qed \end{@proof}}

\usepackage{amssymb}
\usepackage{todonotes}

 \newcommand{\E}{\mathbb{E}}

\newcommand{\ed}{\operatorname{ed}}

\usepackage{tikz}
\usetikzlibrary{arrows,shapes}
\usetikzlibrary{matrix}
\usepackage{algorithm}

\newtheorem{thm}{Theorem}[section]
\newtheorem{defn}[thm]{Definition}
\newtheorem{lem}[thm]{Lemma}

\newtheorem{cor}[thm]{Corollary}

\newtheorem{remark}[thm]{Remark}

\usepackage{subcaption}

\usepackage{hyperref}

\begin{document}

\title{\Large Efficiently Approximating Edit Distance Between Pseudorandom Strings}
\author{William Kuszmaul\thanks{Massachusetts Institute of Technology, Cambridge MA, email: \texttt{kuszmaul@mit.edu}. Supported by an MIT Akamai Fellowship and a Fannie \& John Hertz Foundation Fellowship}}
\date{}

\maketitle







\begin{abstract} \small\baselineskip=9pt
  We present an algorithm for approximating the edit distance $\ed(x,
  y)$ between two strings $x$ and $y$ in time parameterized by the
  degree to which one of the strings $x$ satisfies a natural
  pseudorandomness property. The pseudorandomness model is asymmetric
  in that no requirements are placed on the second string $y$, which
  may be constructed by an adversary with full knowledge of $x$.

  We say that $x$ is \emph{$(p, B)$-pseudorandom} if all pairs $a$ and
  $b$ of disjoint $B$-letter substrings of $x$ satisfy
  $\ed(a, b) \ge pB$. Given parameters $p$ and $B$, our algorithm
  computes the edit distance between a $(p, B)$-pseudorandom string
  $x$ and an arbitrary string $y$ within a factor of $O(1/p)$ in time
  $\tilde{O}(nB)$, with high probability.  If $x$ is generated at
  random, then with high probability it will be $(\Omega(1), O(\log
  n))$-pseudorandom, allowing us to compute $\ed(x, y)$ within a
  constant factor in near linear time. For strings $x$ of varying
  degrees of pseudorandomness, our algorithm offers a continuum of
  runtimes.

  Our algorithm is robust in the sense that it can handle a small
  portion of $x$ being adversarial (i.e., not satisfying the
  pseudorandomness property). In this case, the algorithm incurs an
  additive approximation error proportional to the fraction of $x$
  which behaves maliciously.

  The asymmetry of our pseudorandomness model has particular appeal
  for the case where $x$ is a \emph{source string}, meaning that
  $\ed(x, y)$ will be computed for many strings $y$. Suppose that one
  wishes to achieve an $O(\alpha)$-approximation for each $\ed(x, y)$
  computation, and that $B$ is the smallest block-size for which the
  string $x$ is $(1/\alpha, B)$-pseudorandom. We show that without
  knowing $B$ beforehand, $x$ may be preprocessed in time
  $\tilde{O}(n^{1.5}\sqrt{B})$, so that all future computations of the
  form $\ed(x, y)$ may be $O(\alpha)$-approximated in time
  $\tilde{O}(nB)$. Furthermore, for the special case where only a
  single $\ed(x, y)$ computation will be performed, we show how to
  achieve an $O(\alpha)$-approximation in time
  $\tilde{O}(n^{4/3}B^{2/3})$.

\end{abstract}

\section{Introduction}

The edit distance $\ed(x, y)$ between two strings $x$ and $y$ over an
alphabet $\Sigma$ is the minimum number of insertions, deletions, and
substitutions of characters needed to transform $x$ to $y$. The
textbook dynamic-programming algorithm for edit distance runs in time
$O(n^2)$ \cite{WagnerF74, needleman1970general, vintzyuk1968}, and
conditional lower bounds suggest that no algorithm can do more than a
sub-polynomial factor better \cite{backurs2015edit, hard2} (unless the Strong
Exponential Time Hypothesis fails).

The difficulty of computing edit distance poses a significant
challenge for applications involving large strings. In computational
biology, for example, edit distance is an important tool for comparing
differences between genetic sequences \cite{navarro2001guided, blast,
  needleman1970general}. Rather than relying on the quadratic-time
algorithm, such applications have often resorted to the use of fast
heuristics \cite{gusfield1997algorithms, navarro2001guided,
  heuristicoriginal, blast, ma2002patternhunter}.

On the theoretical side, extensive work has been done towards
algorithms that circumvent the quadratic lower bound, either by
approximating edit distance, or by making assumptions on the input
\cite{ApproxPolyLog, ApproxSubPolyDistortion, approx18, Embedding,
  indyk2001algorithmic, landau1998incremental, smoothcomplexity,
  chakraborty2016streaming2, icalp, haeupler18}.

One of the most notable successes in this direction is the algorithm
of Landau and Myers \cite{landau1998incremental} with runtime
parameterized by the edit distance, computing $\ed(x, y)$ in time $O(n
+ \ed(x, y)^2)$. The algorithm runs in linear time when the edit
distance is small (say, less than $\sqrt{n}$), and offers a continuum
of runtimes for larger edit distances. Another celebrated result is
the approximation algorithm of Andoni et al.  \cite{ApproxPolyLog},
which approximates edit distance within a factor of $(\log
n)^{O(1/\epsilon)}$ in time $O(n^{1 + \epsilon})$. A recent
breakthrough by Chakraborty et al. \cite{approx18} gives the first
approximation algorithm to achieve a sub-logarithmic approximation
ratio in strongly subquadratic time. Their algorithm computes a
constant approximation for edit distance in time $O(n^{12/14})$. It
remains unknown whether a sub-logarithmic approximation for edit
distance can be computed in close to linear time.

It was recently shown by Andoni and Krauthgamer
\cite{smoothcomplexity} that better tradeoffs could be achieved if
certain randomness assumptions are placed on $x$ and $y$. In
particular, they introduce the \emph{smoothed complexity model}, in
which $x$ and $y$ are binary strings which are assumed to be partly
determined by a random process: Given two initial strings $x^\circ$
and $y^\circ$ with longest common subsequence $A$, $x$ and $y$ are
defined by perturbing each letter of $x^\circ$ and $y^\circ$ with
probability $p$, except that $A$ is perturbed in the same way inside
each of the two strings. Given two strings $x$ and $y$ which are
constructed in this manner, the algorithm of \cite{smoothcomplexity}
can be used to compute an $O\left(\frac{1}{\epsilon p} \log
\frac{1}{\epsilon p}\right)$-approximation for $\ed(x, y)$ in time
$O\left(n^{1 + \epsilon} / \sqrt{\ed(x, y)}\right)$. For constant $p$
and small $\epsilon < 1$, this gives a constant approximation in close
to linear (or potentially sublinear) time.

\paragraph*{Our Contribution} We consider an asymmetric randomness model for edit distance, in which
one of the strings $x$ is presumed to satisfy natural randomness
properties, but in which the string $y$ may be constructed by an
adversary (with full knowledge of $x$). We show for a random string $x
\in \Sigma^n$ and an arbitrary string $y \in \Sigma^n$, there is an
approximation algorithm that computes $\ed(x, y)$ up to a constant
factor in near linear time, with high probability.

Rather than requiring that the string $x$ is truly random, our
algorithm also works for any string $x$ satisfying the following
property: There is some constant $c > 1$ such that for any two
disjoint substrings $a$ and $b$ in $x$ of length $c \log n$, the edit
distance between $a$ and $b$ is at least $\log n$.

By relaxing the above property, one can define a natural notion of
pseudorandomness for a string $x$. Specifically, we say that $x$ is
\emph{$(p, B)$-pseudorandom} if all pairs $a$ and $b$ of disjoint
$B$-letter substrings of $x$ satisfies $\ed(a, b) \ge pB$. 

Given parameters $p$ and $B$, our main result is an algorithm that
computes the edit distance between a $(p, B)$-pseudorandom string $x$
and an arbitrary string $y$ within a factor of $O(1/p)$ in time
$\tilde{O}(nB)$, with high probability (Section \ref{secmain}).

An interesting feature of our algorithm is that it allows for
substantial speedups over the quadratic-time algorithm (as well as the
$O(n^{12/14})$-time approximation algorithm of \cite{approx18}) even
for strings $x$ exhibiting relatively limited pseudorandomness. For
example, if $x$ is $(\Omega(1), n^{\epsilon})$-pseudorandom, then we
can recover a constant approximation for $\ed(x, y)$ in time
$\tilde{O}(n^{1 + \epsilon})$.

The asymmetry of our pseudorandomness model seems particularly
appealing for applications in which one string $x$ is treated as a
\emph{source string}, and in which the edit distance $\ed(x, y)$ is
then computed for a large collection of strings $y$. In particular,
one can first determine values of $p$ and $B$ for which $x$ is $(p,
B)$-pseudorandom, and then use our algorithm with parameters $p$ and
$B$ to approximate $\ed(x, y)$ for any string $y$. In Section
\ref{secoblivious}, we present an algorithm for determining the values
of $B$ and $p$ with which to run our approximation algorithm on a
string $x$. In particular, suppose we wish to obtain an
$\alpha$-approximation for each computation $\ed(x, y)$, and that $B$
is the smallest (power-of-two) block size for which the string $x$ is
$(1/\alpha, B)$-pseudorandom. Then without knowing $B$ beforehand, we
can perform an $\tilde{O}(n^{1.5}(\sqrt{B} + \alpha))$-time
preprocessing step on $x$, so that any edit distance computation
$\ed(x, y)$ can later be $O(\alpha)$-approximated in time
$\tilde{O}(nB)$. We also consider the special case where we wish to
perform only a single computation $\ed(x, y)$, but we do not know
beforehand the minimum (power-of-two) block size $B$ for which $x$ is
$(1/\alpha, B)$-pseudorandom; for this case we present an algorithm
which runs in time $\tilde{O}(n^{4/3}B^{2/3})$.

Our approximation algorithm is robust to a small portion of $x$ being
determined by an adversary, rather than satisfying the
pseudorandomness property, and the algorithm does not require us to
know which portions of $x$ have been affected by the adversary. In
this case, the algorithm incurs an additive approximation error
proportional to the fraction of $x$ which has been constructed
adversarially.

\paragraph*{Relationship to the Smoothed Complexity Model}
Interestingly, the smoothed complexity model \cite{smoothcomplexity}
can be viewed as a special case of our pseudorandomness model. In
particular, the string $x$ in the smoothed complexity model will be
$(p, O(\frac{1}{p} \log n))$-pseudorandom with high probability (see
Lemma 2.2 of \cite{smoothcomplexity}). Our algorithm can therefore be
used in the smoothed complexity model to obtain an $\tilde{O}(n /
p)$-time algorithm with approximation ratio $O(1/p)$.

Both our algorithm and the algorithm of \cite{smoothcomplexity} use
substring matching between $x$ and $y$ as an important algorithmic
component, where two substrings are said to match if their edit
distance is relatively small. Substring matching also plays an
important role in a number of the practical heuristics for
computing edit distance, including the widely used PatternHunter tool
\cite{ma2002patternhunter}. The prominent role of the technique within
both our algorithm and the algorithm of \cite{smoothcomplexity} helps
provide theoretical motivation for why the technique has had such
empirical success in practice. The empirical success of the technique
may also be an indicator that real-world inputs satisfy some variant
of $(p, B)$-pseudorandomness. Evaluating to what degree this is the
case would be an important direction of future work.

Beyond the use of substring matching, the techniques used in
\cite{smoothcomplexity} largely differ from ours. In particular, it
seems difficult to extend the techniques used in
\cite{smoothcomplexity} to our pseudorandomness model without using
space (and time) exponential in $p \cdot B$.

\paragraph*{Related Work on Average-Case Analysis for Approximate String Matching}

In the Approximate String Matching Problem, we are given a string $x
\in \Sigma^{n}$, a pattern $p$ of length at most $n$, and a threshold
$k$. The goal is to output all positions $i$ in $x$ such that there
exists some $j > i$ for which $\ed(x_i \cdots x_j, p) < k$. It was
shown by Landau et al. \cite{landau1989fast} that the Approximate
String Matching Problem can be solved in time $O(nk)$ using space
$O(n)$.

Additional space improvements have been presented for the case where
$x$ is selected at random from $\Sigma^{n}$. Notably, an algorithm due
to Myers \cite{myers1986incremental} achieves space $O(k)$ for random
$x$. Improvements in runtime have also been presented. In particular,
an algorithm due to Chang and Lampe \cite{chang1992theoretical}
achieves runtime $O(nk / \sqrt{|\Sigma|})$.

For strings $x$ and $y$ of length $n$, selected from a constant-size
alphabet $\Sigma$, no known algorithm for Approximate String Matching
achieves a strongly subquadratic running time, even for $x$ selected at
random. Consequently, there is relatively little overlap between the
setting considered in this paper and past work on Approximate String
Matching.

\section{Technical Overview}

In this section, we present a brief technical overview of our
approximation algorithm and of our algorithms for parameter detection.

\paragraph*{Approximation Algorithm for $x$ a $(p, B)$-Pseudorandom String} In Section \ref{secmain}, we present our approximation algorithm for computing the edit distance between a $(p, B)$-pseudorandom string $x$ and an arbitrary string $y$. A rough summary of our approximation algorithm is as follows: We break the strings $x$ and $y$ into blocks of size $\Theta(B)$, and define a notion of a block $x^i$ from $x$ \emph{edit matching} with a block $y^j$ from $y$. In order to find a sequence of edits from $x$ to $y$, we perform the following process recursively. We select a random pivot block $x^i$ from the blocks in the middle half of $x$. If the pivot either edit matches with no blocks in $y$, or edit matches with more than one block in $y$, then we consider it to be invalid, and we sample another random pivot until we find a valid one, giving up after $100 \log n$ consecutive failures. If, on the other hand, the pivot edit matches with exactly one block $y^j$ in $y$, then we break the strings at $x^i$ and $y^j$, respectively, and recursively try to align the two strings in each of the two halves. The base case for our algorithm occurs when either the strings $x$ and $y$ differ in size from each other by enough to force their edit distance to be large, or when $|x| = 0$. Our algorithm outputs a non-crossing matching between the blocks of $x$ and the blocks of $y$, where two blocks are only permitted to be matched if they edit match. By considering only sequences of edits which roughly respect the matchings between blocks, we are then able to efficiently find a sequence of edits from $x$ to $y$.

In analyzing the algorithm, we consider an optimal alignment $T$ between $x$ and $y$, and based on $T$ we define two types of subproblems: live subproblems and dead subproblems. \emph{Dead subproblems} are those which are either between vastly differently sized subtrings $u$ and $v$, or which are between substrings $u$ and $v$ which are for the most part not aligned with one-another by $T$. These subproblems prove simple to analyze because any mistakes the algorithm makes on a dead subproblem can be easily accounted for by the inherent misalignment in the subproblem.

The remaining subproblems are known as \emph{live subproblems}, and require a more sophisticated analysis. For each block $x^i$ in $x$, we denote by $\phi(i)$ the index of the block in $y$ closest to where the optimal alignment $T$ maps $x^i$. (i.e., the alignment $T$ maps the block $x^i$ to roughly where $y^j$ sits in $y$.) One of the key insights is that each live subproblem $A$  must satisfy the following property:
\begin{itemize}
\item \textbf{The Pivot Pairing Property: }If the algorithm selects a block $x^i$ as block pivot, then the block $y^{\phi(i)}$ must appear in the substring of $y$ considered by the subproblem.
\end{itemize}
If a pivot block $x^i$ edit matches with its counterpart $y^{\phi(i)}$, then any live subproblem which selects $x^i$ will be guaranteed to perform well. Using the $(p, B)$-pseudorandomness property of $x$, we are able to show that at most $O(\frac{1}{p}\ed(x, y))$ of the letters in $x$ reside in blocks $x^i$ that do not edit match with their counterpart $y^{\phi(i)}$. The blocks containing these letters, which we will call \emph{dirty blocks}, are the blocks which place our algorithm at risk of exhibiting a large approximation ratio. 

The Pivot Pairing Property has the following important consequence for dirty blocks: For each dirty block $x^i$, any live subproblem $A$ which uses $x^i$ as a pivot block, must pair $x^i$ with one of at most two blocks $y^{j_1}$ and $y^{j_2}$, where $j_1$ is the largest value $j_1 < \phi(i)$ for which $x^i$ edit matches with $y^j$, and $j_2$ is the smallest value $j_2 > \phi(i)$ for which $x^i$ edit matches with $y^{j_2}$. Since each dirty block has only two potential partners for it to pair with inside any live subproblem, there are, in total, only $O(\frac{1}{pB}\ed(x, y))$ pairs $(x^i, y^j)$ for which the algorithm is at risk of at some point selecting $x^i$ as a pivot and erroneously matching it with $y^j$. Call such a pair a \emph{dirty pair}.

In order to analyze the distortion of our algorithm, we consider each dirty pair $(x^i, y^j)$, and bound the expected damage to the algorithm's output incurred by the pair. For a given live subproblem $A$ that is at risk of pairing $x^i$ as a pivot with $y^j$, let $r$ denote the number of blocks in the substring of $x$ considered by $A$. We show that because $A$ is a live subproblem, it must contain $\Omega(r)$ valid pivot block options, meaning that $x^i$ has a probability of at most $O(1/r)$ of being selected. If the block $x^i$ is selected as a pivot and is matched with the block $y^j$, then we are able to bound the damage incurred to the algorithm (in terms of number of edits added to the final alignment) by $O(B \cdot |\phi(i) - j|)$, the number of letters which lie between where the optimal alignment $T$ maps $x^i$ and where our algorithm places $x^i$ within $y$. It follows that the expected damage incurred by the pair $(x^i, y^j)$ to our subproblem $A$ is at most
$$O\left(\frac{1}{r} \cdot B \cdot |\phi(i) - j|\right).$$

Using the Pivot Pairing Property, we can show that the size $r$ must be at least $\sigma = \Omega(|\phi(i) - j|)$, since the subproblem must involve both $y^j$ and $y^{\phi(i)}$. Moreover, for each range of the form $R_l = [2^l \sigma, 2^{l + 1} \sigma]$, there can be at most a constant number of subproblems $A$ that contain $x^i$ and have a size $r$ lying in the range. For each range $R_l$, it follows that the expected damage incurred by the pair $(x^i, y^j)$ to live subproblems for which $r \in R_l$ is at most
$$O\left(\frac{1}{2^l \sigma} \cdot B \cdot |\phi(i) - j|\right) = O\left(\frac{1}{2^l}B\right).$$
Summing over $l$, we find that the total expected damage incurred by $(x^i, y^j)$ is at most
$$O\left( \sum_{l \ge 0} \frac{1}{2^l} B\right) = O(B).$$
Summing over the $O(\frac{1}{pB} \ed(x, y))$ dirty pairs $(x^i, y^j)$, the total expected error in our algorithm's estimate for edit distance can be bounded by
$$O\left(\frac{1}{pB} \ed(x, y) \cdot B\right) = O\left(\frac{1}{p}\ed(x, y)\right),$$
as desired.

We break the full presentation and analysis of the algorithm into two parts:
\begin{itemize}
\item The first part is an approximation algorithm for a combinatorial matching problem in which the blocks in the strings $x$ and $y$ have been replaced with abstract objects (essentially characters), and in which a \emph{comparison function} $f(i, j)$ determines whether the blocks $x^i$ and $y^j$ may be matched with one another. The algorithm and analysis for this problem essentially follow the outline above.
\item The second part is a reduction from approximating $\ed(x, y)$ between two strings (one of which is $(p, B)$-pseudorandom) to the combinatorial matching problem. The reduction breaks the string $x$ into parts $x^1, x^2, \ldots$ of size $6B$ and breaks the string $y$ into parts $y^1, y^2, \ldots$ of size $3B$. Roughly speaking, we say that $x^i$ \emph{edit matches} with $y^j$ if there is a $6B$-letter substring $a$ containing $y^j$ for which $\ed(x^i, a) \le O(\frac{1}{p} B)$, and if the same is not true for $y^{j - 1}$.

The reduction encounters two technical subtleties. The first is that, although the strings $x$ and $y$ are the same length initially, they are mapped by the reduction to strings of length $\frac{n}{6B}$ and $\frac{n}{3B}$, respectively. As a result, every alignment in the resulting matching problem must leave at least half of the blocks in $y$ unmatched. This is resolved by defining an asymmetric cost model for alignments in the problem which we reduce to; in particular, a block $y^i$ only pays for being unmatched if the three blocks to each of its immediate left and right are also unmatched. The second subtlety is the handling of the case where some fraction of the blocks in $x$ have been modified by an adversary, rather than satisfying the $(p, B)$-pseudorandomness property within $x$. The reduction is designed to hide this aspect of the problem from the combinatorial matching problem which we reduce to, by ensuring that the optimal cost for the matching problem is given by $O\left(\frac{1}{p} \ed(x, y) + M(x)\right)$, where $M(x)$ is the number of letters in $x$ which lie in adversarially designed blocks. 
\end{itemize}

\paragraph*{Parameter Detection} In Section \ref{secoblivious}, we turn our attention to parameter detection. In particular, consider a source string $x$ and suppose we wish to compute an $O(\alpha)$-approximation for $\ed(x, y)$ for a large number of strings $y$. Define $B$ to be the smallest power-of-two block-size for which $x$ is $(1/\alpha, B)$-pseudorandom. Without knowing $B$ beforehand, we show how to preprocess $x$ in time $\tilde{O}(n^{1.5}(\sqrt{B} + \alpha))$ in order so that all future computations of $\ed(x, y)$ can be $O(\alpha)$-approximated in time $\tilde{O}(nB)$. One way to do to this would be to find the first $B_i = 2^i \alpha$ for which $x$ is $(1/\alpha, B_i)$-pseudorandom, which would take roughly quadratic time per $B_i$. The key observation to improve upon this is that, since we are willing to spend time $\tilde{O}(nB)$ to approximate each value $\ed(x, y)$, if it turns out that $\ed(x, y) \le \sqrt{nB}$, then we may use the $O(n + \ed(x, y)^2)$-time algorithm of \cite{landau1998incremental} in order to exactly compute $\ed(x, y)$ in the desired time bound. By exploiting this trick, we can tolerate our approximation algorithm for $\ed(x, y)$ having additive error of $\sqrt{nB}$. Consequently, rather than finding the first $B_i$ for which $x$ is $(1/\alpha, B_i)$-pseudorandom, we can instead check for the weaker condition that $x$ is $(1/\alpha, B_i)$-pseudorandom with the exception of blocks containing at most $O(\sqrt{nB})$ letters in $x$. Using random sampling, this can be done in time $\tilde{O}(n^{1.5}(\sqrt{B} + \alpha))$, as desired. In a similar manner, for the case where only a single computation $\ed(x, y)$ is to be performed, we obtain an $O(\alpha)$-approximation in time $\tilde{O}(n^{4/3}B^{2/3})$.

\section{Edit Distance Between Pseudorandom Strings}\label{secmain}

Let $0 < p < 1$ and $B \in \mathbb{N}$ be parameters. Throughout the
section, we will use $\Sigma$ to denote the alphabet over which
strings are taken. We will present a randomized algorithm that
computes the edit distance between a $(p, B)$-pseudorandom string $x
\in \Sigma^n$ and an arbitrary string $y \in \Sigma^n$ within a factor
of $O(1/p)$ in time $\tilde{O}(n B)$.

Additionally, our algorithm will be able to handle a small portion of
$x$ behaving maliciously. In the following definition, we define a
quantity $M(x)$ which captures the degree to which $x$ is not $(p,
B)$-pseudorandom.

\begin{defn}
Consider $x \in \Sigma^n$, and break it into blocks $x^1, x^2, \ldots,
x^{\lceil n / 6B \rceil}$ of size $6B$. If $n$ is not a multiple of
$6B$, then pad the final block with null characters so that it is
length $6B$. We say a block $x^i$ in $x$ is \emph{$p$-unique} if each
$B$-letter substring of $x^i$ has edit distance distance at least $pB$
from each $B$-letter substring in $x$ that does not intersect
$x^i$. Define the quantity $M(x)$ (which will often be abbreviated by
$M$) to be the number of letters in $x$ (including null characters
padded on the end) that are not in $p$-unique blocks.
\end{defn}

The goal of the section is to prove the following theorem:

\begin{theorem}
Let $n \in \mathbb{N}$ be a string-length parameter, $p \in (0, 1)$ be
a parameter such that $\frac{1}{p} \in \mathbb{N}$, and $B \in
\mathbb{N}$ be a block-size parameter.

Consider as inputs a string $x \in \Sigma^n$ and an arbitrary string
$y \in \Sigma^n$. Then there exists a randomized algorithm which in
time $\tilde{O}(n \cdot B)$ finds a sequence of $t$ edits between $x$
and $y$, where $t$ is no greater than
\begin{equation}
  O\left(\frac{1}{p}\ed(x, y) + M(x)\right),
  \label{eqresult}
\end{equation}
in expectation.
\label{thmmain}
\end{theorem}

Notice that by repeatedly applying Theorem \ref{thmmain} $O(\log n)$
times, and taking the smallest returned sequence of edits, one can
further ensure that Equation \eqref{eqresult} will hold with high
probability.

An important corollary of Theorem \ref{thmmain} is for the case where
$x$ is fully random.
\begin{cor}
  There is an algorithm which, given a random string $x \in \Sigma^n$,
  and an arbitrary string $y \in \Sigma^n$, finds a sequence of $O(\ed(x, y))$ edits between $x$ and $y$ in time $\tilde{O}(n)$, with high probability.
\end{cor}
\begin{proof}
  This follows from Theorem \ref{thmmain}, and the fact that $x$ will be $(\Omega(1), O(\log n))$-pseudorandom with high probability (See 2.2 of \cite{smoothcomplexity}).
\end{proof}

For the rest of the section, we will assume without loss of generality
that $n$ is a multiple of $6B$. Notice that by padding $x$ with null
characters, one can ensure that the last block is of length $6B$
without increasing $M(x)$, and furthermore, if one places the same
padding on $y$, then the edit distance between the padded strings will
be the same as between the original strings.

For the rest of the section, let $n, p, B, x, y$ be as defined in
Theorem \ref{thmmain}, and assume that $6B \mid n$. In order to prove
Theorem \ref{thmmain}, we will reduce it to a combinatorial problem
which we call the Clean Alignment Problem. Formally, Theorem
\ref{thmmain} can be viewed as presenting a $\tilde{O}(nB)$-time
solution to the Pseudorandom Edit Distance Problem, defined as
follows.

\begin{itemize}
  \item \textbf{The Pseudorandom Edit Distance Problem (PEDP): }Given
    $x \in \Sigma^n$ and an arbitrary $y \in \Sigma^n$, recover a
    sequence of $t$ edits from $x$ to $y$ for $t$ satisfying $\E[t]
    \le O(1/p \cdot \ed(x, y) + M(x))$.
\end{itemize}

We will reduce PEDP to the Clean Alignment Problem, defined
below. Specifically, any algorithm for the Clean Alignment Problem
which runs in time $T(n)$ will imply an algorithm for PEDP which runs in
time $\tilde{O}(T(n) \cdot B)$.
 
 \begin{itemize}
\item \textbf{The Clean Alignment Problem: }Consider strings
  $x$ and $y$ of possibly different lengths. Let $f: [|x|] \times [|y|]
  \rightarrow \{0, 1\}$ be a \emph{comparison function} for $x$ and
  $y$. If $f(i, j) = 1$, we say that $x_i$ \emph{edit matches} with
  $y_j$.

  An alignment $\mathcal{A}$ (i.e., a non-crossing matching) between
  the letters of $x$ and the letters of $y$ is said to be an \emph{edit-matching alignment} 
  if every edge $(x_i, y_j) \in \mathcal{A}$ satisfies $f(i, j) = 1$. The cost of
  an alignment is defined asymmetrically: The \emph{$x$-portion cost}
  of $\mathcal{A}$ is the number of letters in $x$ to be unmatched
  by $\mathcal{A}$; the $y$-portion cost is the number of $i$ for
  which $y_i, y_{i + 1}, \ldots, y_{i + 6}$ are all unmatched by
  $\mathcal{A}$. The \emph{cost} of $\mathcal{A}$ is the sum of the $x$-portion and $y$-portion costs.

  An edit-matching alignment is said to be \emph{clean} if for all
  $(x_i, y_j) \in \mathcal{A}$, $x_i$ is the only letter in $x$ which edit
  matches with $y_j$, and $y_j$ is the only letter in $y$ which edit
  matches with $x_i$.

  Given $x$, $y$, and $f$, the goal of the \emph{Clean
    Alignment Problem} is to recover a (not necessarily clean)
  edit-matching alignment between $x$ and $y$ whose expected cost is
  within a constant factor of the optimal cost of a clean
  edit-matching alignment.
\end{itemize}

We remark that the difficulty of the Clean Alignment Problem comes
from the fact that the function $f$ in the Clean Alignment Problem
needs not satisfy any natural properties (such as e.g., some form of
transitivity). Thus the only way to determine whether $f(i, j) = 1$
for two letters $x_i$ and $y_j$, is to explicitly evaluate $f(i,
j)$. In order for an algorithm to run in time $\tilde{O}(n)$, as
desired, it must very selectively determine for which values of $i$
and $j$ to query $f(i, j)$.

The reader may notice that the Clean Alignment Problem is completely
determined by $f$, $|x|$, and $|y|$. That is, the actual contents of
$x$ and $y$ are irrelevant. Nonetheless, for ease of presentation, we
include $x$ and $y$ as inputs. For examples of a comparison function
$f$, and of clean edit-matching alignments, see Figure
\ref{figcleanalignment}.

\begin{figure*}[h]
\begin{subfigure}[b]{1 \textwidth}
\begin{center} \begin{tikzpicture}
\matrix[name=m,nodes={anchor=base west},matrix of nodes,row sep={2em,between origins},column sep=\tabcolsep,ampersand replacement=\&]{
$y$: \&  $y_1$ \& $y_2$ \& $y_3$ \& $y_4$ \& $y_5$ \& $y_6$ \& $y_7$ \& $y_8$ \& $y_9$ \& $y_{10}$ \& $y_{11}$ \& $y_{12}$ \& $y_{13}$ \& $y_{14}$ \& $y_{15}$ \\
\&    \&   \&   \&   \&   \&   \\ 
$x$:  \&  $x_1$ \& $x_2$ \& $x_3$ \& $x_4$ \& $x_5$ \& $x_6$ \& $x_7$ \& $x_8$ \& $x_9$ \& $x_{10}$ \& $x_{11}$ \& $x_{12}$ \& $x_{13}$ \& $x_{14}$ \& $x_{15}$ \\
};
\draw[-] (m-3-3) -- (m-1-5);
\draw[color = blue] (m-3-4) -- (m-1-10);
\draw[-] (m-3-7) -- (m-1-5);
\draw[color = blue] (m-3-9) -- (m-1-8);
\draw[-] (m-3-10) -- (m-1-7);
\draw[-] (m-3-10) -- (m-1-11);
\draw[color = blue] (m-3-13) -- (m-1-13);
\end{tikzpicture} \end{center}
\caption{An edge is drawn between $x_i$ and $y_j$ if $f(i, j) = 1$. Edges which are eligible to be used in a clean edit-matching alignment are colored blue. This occurs only when the edge is not incident to any other edges.}
\end{subfigure}

\begin{subfigure}[b]{1 \textwidth}
\begin{center} \begin{tikzpicture}
\matrix[name=m,nodes={anchor=base west},matrix of nodes,row sep={2em,between origins},column sep=\tabcolsep,ampersand replacement=\&]{
$y$: \&  $y_1$ \& $y_2$ \& $y_3$ \& $y_4$ \& $y_5$ \& $y_6$ \& $y_7$ \& $y_8$ \& $y_9$ \& $y_{10}$ \& $y_{11}$ \& $y_{12}$ \& $y_{13}$ \& $y_{14}$ \& $y_{15}$ \\
\&    \&   \&   \&   \&   \&   \\ 
$x$:  \&  $x_1$ \& $x_2$ \& $x_3$ \& $x_4$ \& $x_5$ \& $x_6$ \& $x_7$ \& $x_8$ \& $x_9$ \& $x_{10}$ \& $x_{11}$ \& $x_{12}$ \& $x_{13}$ \& $x_{14}$ \& $x_{15}$ \\
};
\draw[-] (m-3-4) -- (m-1-10);
\draw[-] (m-3-13) -- (m-1-13);
\end{tikzpicture} \end{center}
\caption{A clean edit-matching alignment of cost 15. The $x$-portion
  cost is $13$ because $13$ nodes are left unmatched in $x$. The
  $y$-portion cost is $2$ because there are two edge-free contiguous
  intervals in $y$ of size seven.}
\end{subfigure}

\begin{subfigure}[b]{1 \textwidth}
\begin{center} \begin{tikzpicture}
\matrix[name=m,nodes={anchor=base west},matrix of nodes,row sep={2em,between origins},column sep=\tabcolsep,ampersand replacement=\&]{
$y$: \&  $y_1$ \& $y_2$ \& $y_3$ \& $y_4$ \& $y_5$ \& $y_6$ \& $y_7$ \& $y_8$ \& $y_9$ \& $y_{10}$ \& $y_{11}$ \& $y_{12}$ \& $y_{13}$ \& $y_{14}$ \& $y_{15}$ \\
\&    \&   \&   \&   \&   \&   \\ 
$x$:  \&  $x_1$ \& $x_2$ \& $x_3$ \& $x_4$ \& $x_5$ \& $x_6$ \& $x_7$ \& $x_8$ \& $x_9$ \& $x_{10}$ \& $x_{11}$ \& $x_{12}$ \& $x_{13}$ \& $x_{14}$ \& $x_{15}$ \\
};
\draw[-] (m-3-9) -- (m-1-8);
\draw[-] (m-3-13) -- (m-1-13);
\end{tikzpicture} \end{center}
\caption{A clean edit-matching alignment of cost 14. The $x$-portion
  cost is $13$ because $13$ nodes are left unmatched in $x$. The
  $y$-portion cost is $1$ because one edge-free contiguous interval in
  $y$ of size seven.}
\end{subfigure}
\caption{Two examples of clean edit-matching alignments.}
\label{figcleanalignment}
\end{figure*}
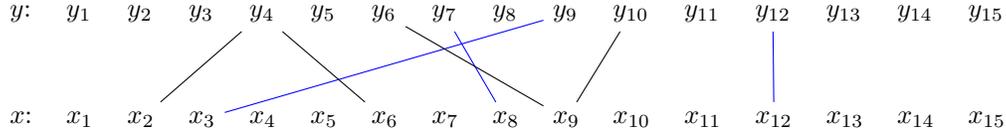
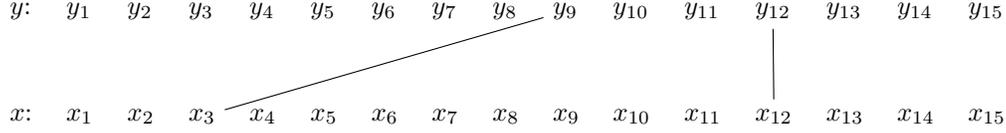
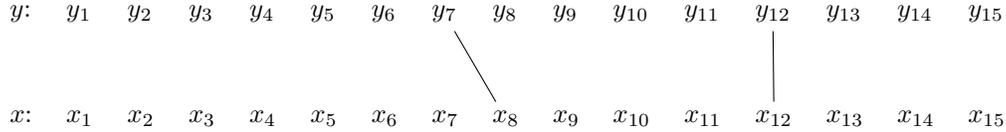

The remainder of the section is outlined as follows. In Subsection \ref{subsecmatching} we present an algorithm for the Clean Alignment Problem which runs in near linear time. In Subsection \ref{subsecreduction}, we then present a reduction from PDEP to the Clean Alignment Problem, completing the proof of Theorem \ref{thmmain}.

\subsection{The Clean Alignment Problem}\label{subsecmatching}

In this section we present an approximation algorithm for the Clean
Alignment Problem. The algorithm is given by Algorithm
\ref{alg:approxmatching}, and the key properties of the algorithm are
stated in Theorem \ref{thmcomparisonbased}.

For this section, consider $u, v \in \Sigma^{\le n}$ (meaning they are
of length at most $n$ over the alphabet $\Sigma$), and let $\Pi(u, v)$ denote the edit-matching
alignment given by the output of Algorithm \ref{alg:approxmatching} on
$u$ and $v$. Moreover, let $T$ be an optimal clean edit-matching
alignment between $u$ and $v$.

\begin{thm}
Algorithm \ref{alg:approxmatching} can be evaluated in time
$\tilde{O}(n)$. Moreover, if $t$ is the cost of $T$, then the the
expected number of edges in $T \setminus \Pi(u, v)$ is at most $O(t)$,
in expectation.
\label{thmcomparisonbased}
\end{thm}

Note that removing an edge from an edit-matching alignment can
increase each of the $u$-portion and the $v$-portions of the cost by
at most $1$ each. Thus Theorem \ref{thmcomparisonbased} implies that
Algorithm \ref{alg:approxmatching} achieves constant multiplicative
error, in expectation, for the Clean Alignment Problem.

\begin{cor}
  Algorithm \ref{alg:approxmatching} solves the Clean Alignment
  Problem in time $\tilde{O}(n)$.
  \label{cormatching}
\end{cor}

\begin{algorithm*}
\caption{Approximation Algorithm for The Clean Alignment
  Problem \label{alg:approxmatching}} Input: Strings $u, v \in
\Sigma^{\le n}$, comparison function
$f$. \\ Output: An edit-matching alignment between $u$ and $v$.
\begin{enumerate}
\item If $|v| \ge 8|u| + 12$ or  $|u| \ge 2|v|$, return $\emptyset$.
\item If $|u| = 0$, return $\emptyset$.
\item For $100 \log n$ attempts:
  \begin{enumerate}
  \item Randomly sample a letter $u_i$ satisfying $\frac{1}{4}|u| \le i \le \Big\lceil \frac{3}{4}|u| \Big\rceil $.
  \item In time $O(|v|)$, construct the set $S = \{j \mid f(i, j) = 1\}$.
  \item If $|S|$ is of size one, containing a single element $j$, then:
    \begin{enumerate}
    \item Initialize an output set $O$ containing the single edge $(u_i, v_j)$
    \item Recurse on $u_1 \cdots u_{i  -1}$ and $v_1 \cdots v_{j - 1}$, and add the resulting edges to $O$.
    \item Recurse on $u_{i + 1} \cdots u_{|u|}$ and $v_{j + 1} \cdots v_{|v|}$, and add the resulting edges to $O$.
    \item Return $O$.
    \end{enumerate}
  \end{enumerate}
\item Return $\emptyset$.
\end{enumerate}
\end{algorithm*}

\begin{remark}
Algorithm \ref{alg:approxmatching} is motivated in part by our
previous work on efficiently embedding Ulam distance (edit distance
over permutations) into Hamming space with (optimal) expected
distortion $O(\log n)$ \cite{icalp}. In particular, a similar analysis
to that of \cite{icalp} could be employed in order to prove that
Algorithm~\ref{alg:approxmatching} has approximation ratio $O(\log
n)$, in expectation. The key difference between Algorithm~\ref{alg:approxmatching} and the embedding of \cite{icalp} is that
Algorithm \ref{alg:approxmatching} is permitted to try $O(\log n)$
options for pivots in order to find a ``good'' one, while the
embedding of \cite{icalp} (by virtue of being oblivious to $y$) is
forced to blindly select a pivot. Surprisingly, this minor difference
allows for the more sophisticated analysis of Algorithm
\ref{alg:approxmatching} below, resulting in a constant expected
approximation ratio.
\end{remark}

Before analyzing Algorithm \ref{alg:approxmatching}, we first
introduce several conventions for talking about the algorithm's
subproblems.

\begin{defn}
Consider a subproblem $A$ which involves a substring
$u_i \cdots u_j$ of $u$ and a substring $v_k \cdots v_l$ of $v$. We
call $u_i \cdots u_j$ the \emph{$u$-chunk} of $A$ and $v_k \cdots v_l$
the \emph{$v$-chunk} of $A$. 
\end{defn}

\begin{defn}
We say that the \emph{pivot options} for
$A$ are the $u_i$'s in the middle half of $A$'s $u$-chunk (i.e., $\frac{1}{4}|u| \le i \le \lceil \frac{3}{4}|u|\rceil$). 
Moreover, we say a pivot option $u_i$ is a \emph{valid pivot} if $f(i, j) = 1$ for exactly one $v_j$ in the $v$-chunk of $A$.
\end{defn}

\begin{defn}
For any subproblem $A$, we define $T \cap A$ to be the set of edges in $T$
 which go between the $u$-chunk and $v$-chunk of $A$. That is, $T \cap A$
  is formally $T \cap \{(u_i, v_j) \mid  u_i \in a, v_j \in b\}$, where $a$ is the 
  $u$-chunk and $b$ is the $v$-chunk of $A$. 
\end{defn}

We classify subproblems into two types, live subproblems and dead subproblems. Dead subproblems will act essentially as base cases in the analysis of the algorithm. Live subproblems, on the other hand, will prove far more interesting to analyze.

\begin{defn}
  We call a subproblem with $u$-chunk $a$ and $v$-chunk $b$ \emph{dead} if at least one of the following holds:
  \begin{itemize}
  \item  $|b| \ge 8|a| + 12$ or $|a| \ge 2|b|$.
  \item $|a| = 0$.
  \item At least $|a| / 10$ elements of $a$ are not matched by $T$ to
    an element of $b$.
  \end{itemize}
  Otherwise, we call the subproblem \emph{live}.
\end{defn}

One important property of live subproblems is that they have many valid options for the algorithm to select as a pivot.

\begin{lem}
For a live subproblem $A$, at least $4/5$ of $A$'s pivot options are valid pivots.
\label{lemvalidpivots}
\end{lem}
\begin{proof}
Let $a$ be the $u$-chunk and $b$ be the $v$-chunk of $A$. By the definition of a live subproblem, fewer than $1/10$ of the elements in $a$ fail to be matched by $T$ to an element of $b$. Since at least $1/2$ of the elements of $a$ are pivot options, at least $.8$ of the pivot options must be matched by $T$ to an element of $b$. Because $T$ is a clean edit-matching alignment, any pivot-option which is matched by $T$ to an element of $b$ must be a valid pivot.
\end{proof}

When the algorithm decides on a pivot on which to split a subproblem $A$, the quality of that pivot can be captured by what we refer to as the cost increase of $A$: 
\begin{defn}
 For a live subproblem $A$ with subproblems $B$ and $C$, we define the
 \emph{cost increase} $c(A)$ of $A$ as follows. If the algorithm fails to select a pivot, or selects a 
 pivot $u_i$ in $A$ and matches it with some $v_j$ for which $(u_i, v_j) \in T$, 
 then the cost increase is $c(A) = 0$. Otherwise the cost increase is $c(A) = |T \cap A| - |T \cap B| - |T \cap C|$. That is, $c(A)$ is the number of edges in $T$ that are cut by the selection of the subproblems $B$ and $C$.
\end{defn}

The key technical challenge in the section will be to bound the cost
increases summed over all live subproblems. We begin by presenting three
properties of live subproblems which will be useful in the
analysis. These properties rely on the notion of a \emph{position map}, defined as follows.

\begin{defn}
 The position map $\phi : [u] \rightarrow [v]$ maps positions in $u$
 to the position in $v$ to which they are assigned by the alignment
 $T$. Formally, if $u_l$ is edit matched to some $v_k$ by $T$, we
 define $\phi(l) = k$, and otherwise we recursively define $\phi(l) =
 \phi(l - 1)$ (or $\phi(l) = 1$ if $l = 1$).
\end{defn}

\begin{lem}  
  Suppose a live subproblem $A$ selects some $u_i$ as a pivot and
  matches $u_i$ to some $v_j$. Then the following three properties
  hold:
  \begin{enumerate}
  \item The cost increase of $A$ will be at most $|\phi(i) - j|
    + 1$.
  \item The letter $v_{\phi(i)}$ is contained in $A$'s $v$-chunk.
  \item The size of $A$'s $u$-chunk is at least
    $$\max\left( 1, \frac{|\phi(u_i) - j| - 11}{8} \right).$$
  \end{enumerate}
  \label{lemproperties}
\end{lem}
\begin{proof}
 We begin by proving the first property. Since every edge $(u_r, v_s)
 \in T \cap A$ either satisfies $r \ge i$ and $s \ge \phi(i)$ or $r <
 i$ and $s < \phi(i)$, all edges $(u_r, v_s)$ in the set
 $$(T \cap A) \setminus ((T \cap B) \cup (T \cap C))$$ must satisfy
 the property that $s$ is between $j$ and $\phi(i)$ inclusive. The
 number $c(A)$ of such edges can therefore be at most
 $$|\phi(i) - j| + 1.$$ 

 To prove the second property, suppose for contradiction that
 $v_{\phi(i)}$ is not contained in $A$'s $v$-chunk. Then because
 $\phi$ is weakly increasing, it must be that $v_{\phi(k)}$ is not in
 $A$'s $v$-chunk for at least $1/4$ of the $u_{k}$'s in $A$'s
 $u$-chunk. But this would mean that the subproblem $A$ is dead, a
 contradiction.

 Finally, to prove the third property, notice that the second property
 forces the size of $A$'s $v$-chunk to be at least $|\phi(u_i) - j| +
 1$. By the definition of a live subproblem, it follows that the size
 of $A$'s $u$-chunk must be at least
 $$\max\left( 1, \frac{|\phi(u_i) - j_q| - 11}{8} \right),$$
 as desired.
\end{proof}

Next we present the key technical lemma of the section, in which we
bound the sum of the cost increases of all live subproblems.

 \begin{lem}
Let $t$ be the cost of the optimal clean edit-matching $T$. Let $S$ be the sum of the cost increases over all live subproblems. Then $\E[S] \le O(t)$. 
\label{lemmain}
 \end{lem}
\begin{proof}
Rather than attributing the cost increase $c(A)$ to the subproblem
$A$, we will instead attribute the cost to $u_i$, where $u_i$ is the
pivot selected within the subproblem. If a letter $u_i$ is edit
matched by $T$, then no non-zero cost increase will ever be attributed
to it, since the only way $u_i$ can be selected as a pivot is if it is
to be correctly matched with $v_{\phi(i)}$. (Recall, in particular,
that $T$ is a \emph{clean} edit-matching alignment.) If, on the other
hand, $u_i$ is unmatched by $T$, then we will prove that the expected
sum of cost increases attributed to it is $O(1)$. Since there are
$O(t)$ such $u_i$'s (recall $t$ is the cost of $T$), this will
complete the proof.
    
For the rest of the proof, consider some $u_i$ not matched by $T$. By
the second part of Lemma \ref{lemproperties}, in order for a
subproblem $A$ to attribute its cost increase to $u_i$, it must be
that $u_i$ is in $A$'s $u$-chunk and $v_{\phi(i)}$ is in $A$'s
$v$-chunk. Moreover, because $v_{\phi(i)}$ is in $A$'s $v$-chunk,
there are at most two options for the letter $v_j$ which the algorithm
pairs $u_i$ with -- namely, $v_{j}$ must either the first letter
$v_{j_1}$ to $v_{\phi(i)}$'s left for which $f(i, j_1) = 1$, or it
must be the first such letter $v_{j_2}$ to $v_{\phi(i)}$'s right. We
say that the cost increase of $A$ is attributed to the $(u_i,
v_{j_1})$ pair if $j = j_1$, and is attributed to the $(u_i, v_{j_2})$
pair if $j = j_2$. For the rest of the proof, we will consider $j_q$
for some $q \in \{1, 2\}$, and show that the total cost increases
attributed to the pair $(u_i, v_{j_q})$ is at most $O(1)$, in
expectation.

Define $R_l$ to be the range $[(10/9)^l \sigma, (10/9)^{l +
    1}\sigma]$, for $l \ge 0$ and for
        $$\sigma = \max \left( 1, \frac{|\phi(i) - j_q| - 11}{8} \right).$$ 

For each $R_l$, there is at most one live subproblem $A$ whose
$u$-chunk both has a size in the range $R_l$ and contains
$u_i$. Moreover, by Lemma \ref{lemvalidpivots}, the probability that
the subproblem picks $u_i$ as a pivot is at most
$O\left(\frac{1}{(10/9)^l\sigma}\right)$. Since any live subproblem
which attributes its cost increase to the pair $(u_i, v_{j_q})$ must
have its $u$-chunk be of a size at least $\sigma$ (by the third part
of Lemma \ref{lemproperties}), the $u$-chunk's size must be in one of
the ranges $S_l$. It follows that the expected number of subproblems
which attribute their cost increase to the pair $(u_i, v_{j_q})$ is at
most
    $$O\left(\sum_{l = 0}^\infty \frac{1}{(10/9)^l \sigma} \right) \le
O\left( \frac{1}{\sigma}\right) = O\left(\frac{1}{|\phi(i) -
  j_q|}\right).$$ By the first part of Lemma \ref{lemproperties}, each
such subproblem has cost increase $O(|\phi(i) - j_q|)$. It follows
that the total cost attributed to the pair $(u_i, v_{j_q})$ is at most
$O(1)$, in expectation, completing the proof.
\end{proof}

Lemma \ref{lemmain} shows that the algorithm behaves well on live subproblems. In order to show 
that the same is true for dead subproblems, we introduce the notions of induced and surplus costs. 

\begin{defn}
The \emph{induced cost} of a subproblem $A$ with $u$-chunk $a$ and
$v$-chunk $b$ is the cost of the clean edit-matching alignment between
$a$ and $b$ obtained by restricting $T$ to the two substrings (i.e.,
the alignment with edges $T \cap A$).

The \emph{surplus cost} of the subproblem is the number of edges in
$T$ between $a$ and $b$ that do not appear in the algorithm's output
(i.e., $\Big|(T \cap A) \setminus \Pi(u, v) \Big|$).
\end{defn}

The ratio of a subproblem's induced cost to its surplus cost is a natural measure of how well $\Pi$ performs on the subproblem. 
The next lemma shows that, in this regard, Algorithm \ref{alg:approxmatching} behaves well on dead subproblems.

\begin{lem}
 The surplus cost of the return value of a dead subproblem is at most
 $O(c)$, where $c$ is the induced cost of the subproblem.
 \label{lemdeadsubproblems}
\end{lem}
\begin{proof}
  Let $a$ be the $u$-chunk and $b$ be the $v$-chunk of a dead
  subproblem.
  
  If $|a| = 0$, then the surplus cost of the subproblem is $0$, and is
  thus trivially at most $O(c)$.
  
  In the remaining cases, we will show that the induced cost $c$ is at
  least $\Omega(|a|)$, making the surplus cost $O(c)$ trivially.

  If at least $1/10$ of the elements of $a$ are not matched by $T$
  to an element of $b$, then the induced cost $c$ is at least $|a| /
  10$, as desired. Similarly, if $|a| \ge 2|b|$, then at least half of
  the elements of $a$ will not be matched by an $T$, implying that the
  induced cost $c$ is at least $|a| / 2$.

  If $|b| \ge 8 |a| + 12$, then the $b$-portion cost of any
  edit-matching alignment between $a$ and $b$ must be at least
  $|a|$. In particular, the number of edges in such an alignment can
  be at most $|a|$, and each edge can be contained in at most $7$
  contiguous substrings of length $7$ in $b$. As a result, at most
  $7|a|$ of the contiguous substrings of length $7$ in $b$ can contain a letter incident to an edge. Since $b$ contains at least $8|a|$ contiguous
  substrings of length $7$, the $b$-portion cost of any edit matching
  alignment must be at least $|a|$.
\end{proof}

Combining Lemmas \ref{lemmain} and \ref{lemdeadsubproblems}, we are now prepared to prove Theorem \ref{thmcomparisonbased}.

\begin{proof}[Proof of Theorem \ref{thmcomparisonbased}]
  It is straightforward to prove that Algorithm
  \ref{alg:approxmatching} runs in time $\tilde{O}(n)$. For the rest
  of the proof we focus on the cost of the edit-matching alignment
  $\Pi$ given by the algorithm. Specifically, we wish to show that the
  number of edges in $T \setminus \Pi$ is at most $O(t)$ in
  expectation.

 Consider a live subproblem $A$. By Lemma \ref{lemvalidpivots}, each of the $100 \log n$ attempts at  
 selecting a pivot will succeed with probability at least $.8$, and
 with probability greater than $1 - \frac{1}{n^{100}}$, a pivot will
 end up being selected. Since there are $O(n)$ subproblems in the
 entire algorithm, with probability at least $1 -
 O\left(\frac{1}{n^{99}}\right)$, every live subproblem will succeed
 in one of its $100 \log n$ attempts to spawn recursive
 children. Noting that the number of edges in $T \setminus
 \Pi(u, v)$ is at most $O(n)$, it follows that the cases where
 some live subproblem fails to spawn children have negligible impact
 on $\E[|T \setminus \Pi(u, v)|]$. For the rest of the proof,
 we will condition on all live subproblems succeeding at spawning
 children.
 
  Recall that $c(A)$ denotes the cost increase of a live subproblem $A$.  
  For a live subproblem $A$ with children subproblems $B$ and $C$, we also define the
  \emph{induced cost increase} $c'(A)$ to be the induced cost 
  of $B$ plus the induced cost of $C$ minus the induced cost of $A$. Note that 
  $c'(A) \le 2c(A)$, since the increases in induced cost between $A$ and its 
  subproblems $B$ and $C$ can be attributed to the removals of edges. (Recall that 
  the removal of an edge from an edit-matching alignment increases each of the $u$-portion 
  cost and $v$-portion cost by at most one.)

  Let $I$ be the sum of the cost increases over all live subproblems,
  and $I'$ be the sum of the induced cost increases over all live
  subproblems. Note that the induced cost of the root subproblem is
  precisely the cost $t$ of the optimal clean edit-matching alignment
  $T$. Call a dead subproblem \emph{fresh} if its parent subproblem is
  live. Then the sum of the induced costs of all fresh dead
  subproblems is exactly $t + I'$.  By Lemma \ref{lemdeadsubproblems},
  it follows that the sum $S$ of the surplus costs of all fresh dead
  subproblems is at most $O(t + I')$. Since $I' \le 2I$, we have that
  $S \le O(t + I)$.

  Observe that $|T \setminus \Pi(u, v)|$ is precisely $I$, the sum of
  the cost increases over all live subproblems, plus $S$, the sum of
  the surplus costs over all fresh dead subproblems. Thus
  $$|T \setminus \Pi(u, v)| \le O(I + S) \le O(t + I).$$ 
  By Lemma \ref{lemmain}, $\E[I] \le O(t)$, completing the proof.
  
\end{proof}

\subsection{Reduction to The Clean Alignment Problem}\label{subsecreduction}

In this section, we prove the following theorem:
\begin{thm}
    \label{thmreduction}
  Suppose there is an algorithm for the Clean Alignment
  Problem that runs in time $\tilde{O}(n)$. Then there is an
  algorithm for the Pseudorandom Edit Distance Problem which runs in
  time $\tilde{O}(nB)$.
\end{thm}

Before proving Theorem \ref{thmreduction}, we combine it with
Corollary \ref{cormatching} in order to complete the proof of Theorem
\ref{thmmain}, establishing an efficient solution to the Pseudorandom
Edit Distance Problem.

\begin{proof}[Proof of Theorem \ref{thmmain}]
  By Theorem \ref{thmreduction}, it suffices to present a solution to
  the Clean Alignment Problem which runs in time
  $\tilde{O}(n)$. Corollary \ref{cormatching} gives such a solution.
\end{proof}

The reduction which we will use to prove Theorem \ref{thmreduction}
can be easily stated: Recall that $x^1, x^2, \ldots, x^{n / 6B}$
breaks $x$ into blocks of length $6B$. Let $y^1, \ldots, y^{n / 3B}$
be $y$ broken into blocks of size $3B$.

Let $r$ be the starting index $r = (j - 1) \cdot 3B + 1$ of the block
$y^{j - 1}$. We say that $x^i$ \emph{partially edit matches} with
$y^j$ if there is some $s = r + t \cdot pB / 100$ with $t \in \{0,
\ldots, 300/p - 1\}$ such that $y_s \cdots y_{s + 6B - 1}$ has edit
distance no greater than $pB / 8$ from $x^i$. That is, $x^i$ partially
edit matches $y^j$ if there is a $6B$-letter substring $a$ of $y$ such
that (1) $a$'s start position appears in $y^{j - 1}$; (2) $a$'s
zero-indexed start position is a multiple of $pB / 100$; and (3)
$\ed(a, x^i) \le pB / 8$. Note that $x^i$ can only partially edit match
with $y^j$ for values of $j \ge 2$.

We say that $x^i$ \emph{(fully) edit matches} with $y^j$ if $x^i$
partially matches with $y^j$ but does not partially match with $y^{j -
  1}$.

In order to prove Theorem \ref{thmreduction}, we will show that in
order to solve the Pseudorandom Edit Distance Problem on $x$ and $y$,
it suffices to solve the Clean Alignment Problem on $(x^1,
\ldots, x^{n/ 6B})$ and $(y^1, \ldots, y^{n / 3B})$, with $f(i, j)$
defined to indicate whether $x^i$ and $y^j$ edit match.

We begin by deriving several important properties of this definition
of edit matching. Recall that a block $x^i$ is $p$-unique if its
$B$-letter substrings are all of edit distance at least $pB$ from the
$B$-letter substrings not intersecting $x^i$.

\begin{lem}
  Suppose $x^i$ and $x^j$ both partially edit match with $y^k$. Then
  neither $x^i$ nor $x^j$ are $p$-unique.
  \label{lemnotpunique}
\end{lem}
\begin{proof}
  Since $x^i$ partially edit matches with $y^k$, there is a substring
  $a$ of $x^i$ which is within $pB/8$ edits of $y^k$ (because $x^i$
  must be within $pB / 8$ edits of a string containing $y^k$)
  . Similarly, there is a substring $b$ of $x^j$ which is within
  $pB/8$ edits of $y^k$. By the triangle inequality, $\ed(a, b) \le
  pB/4$. It follows that the first $B$ letters of $a$ are within
  $pB/2$ edits of the first $B$ letters of $b$. (Indeed, if we
  consider an optimal sequence of edits from $a$ to $b$ restricted to
  the first $B$ letters of $a$, then that sequence of edits must
  result in a string of length no more than $\ed(a, b) \le pB / 4$ away from $B$.)
  Therefore, neither $x^i$ nor $x^j$ are $p$-unique.
\end{proof}

For the rest of the subsection, fix an optimal substitution-free
sequence of edits $\mathcal{A}$ between $x$ and $y$. (That is,
$\mathcal{A}$ consists only of insertions and deletions.) We will
sometimes view $\mathcal{A}$ as a sequence of edits, and other times
consider it as a non-crossing matching between the letters in $x$ and
those in $y$ such that each edge must be between two letters of the
same value; the insertions and deletions correspond with letters which
are left unmatched. As such, we will often refer to $\mathcal{A}$ as
an alignment, rather than a sequence of edits.

Assign each edit in $\mathcal{A}$ to a block $x^i$ in which it occurs
(deciding boundary cases arbitrarily), and define $c_i$ to be the
number of edits attributed to each block $x^i$. Notice, in particular,
that $\sum_i c_i = \ed(x, y)$, that each $x^i$ is transformed by
$\mathcal{A}$ into some substring of $y$ via $c_i$ edits, and that the
concatenation of the resulting substrings of $y$ forms all of
$y$. Similarly, attribute to the blocks of $y$ costs $d_i$ such that
$\sum_i d_i = \ed(x, y)$ and such that each $y^i$ is transformed by
$d_i$ edits into some substring of $x$. We will use these values $c_i$
and $d_i$ to aid us in our analysis for the rest of the section.

We continue by stating another important property of edit matching
blocks. We call $x^j$ \emph{polygamous} if it partially edit matches
with some $y^k$ and $y^l$ for which $k$ and $l$ differ by more than
one. Otherwise, $x^j$ is said to be \emph{monogamous}, even if it
fails to partially edit match with any $y^k$. The next lemma bounds
the number of polygamous $x^j$'s.

\begin{lem}
The number of letters in polygamous $x^i$'s is at most
$O\left(\frac{1}{p}\ed(x, y) + M(x)\right)$.
\label{lempolygamous}
\end{lem}
\begin{proof}
  Since only $M(x)$ letters in $x$ reside in non-$p$-unique blocks, it
  suffices to show that the number of letters in polygamous $p$-unique
  $x^i$'s is at most $O(\frac{1}{p}\ed(x, y))$.

  Since $\sum_i c_i = \sum_i d_i = \ed(x, y)$, at most $O(\frac{1}{pB}\ed(x, y))$ of the $c_i$'s and $d_i$'s can be
  greater than $pB / 100$. By Lemma \ref{lemnotpunique}, each $y^k$
  can be partially edit matched with at most one $p$-unique
  $x^j$. Hence the number of letters in $p$-unique $x^j$'s which
  either satisfy $c_j > pB / 100$ or partially match with some $y^k$
  for which $d_k > pB / 100$, is at most $O(\frac{1}{pB} \ed(x, y))$.

  It therefore suffices to bound the number of $p$-unique $x^j$'s satisfying
  the following property: $x^j$ is polygamous with some $y^k$ and
  $y^l$ such that $c_j, d_k, d_l \le pB / 100$. We will show that no
  such $x^j$'s exist, completing the proof.

  Suppose for contradiction that such an $x^j$ exists, and assume
  without loss of generality that $k < l$. Let $\hat{y}^k$ be the
  first third of the block $y^k$, and $\hat{y}^l$ be the final third
  of the block $y^l$. There must be a substring $w_1$ of $x$ which is
  aligned by $\mathcal{A}$ to $\hat{y}^k$ in at most $pB / 100$
  edits. Similarly there must be a substring $w_2$ of $x$ which is
  aligned by $\mathcal{A}$ to $\hat{y}^l$ in at most $pB/100$
  edits. Define $w_1$ and $w_2$ to be minimal such substrings (meaning
  the first and final letters of each of $w_1$ and $w_2$,
  respectively, must be incident to edges in $\mathcal{A}$ mapping to
  letters in $\hat{y}^k$ and $\hat{y}^l$, respectively). Because
  $\hat{y}^k$ and $\hat{y}^l$ are separated by $7B$ letters, $w_1$ and
  $w_2$ cannot both intersect $x^i$, since this would cause $c_i$ to
  be at least $B$. Let $t \in \{1, 2\}$ be such that $w_t$ does not
  intersect $x^i$. Because $x^i$ partially edit matches with $y^k$ and
  $y^l$, it must be that $\hat{y}^k$ and $\hat{y}^l$ are each within
  $pB/8$ edits of some substrings $b_1$ and $b_2$ of $x^i$. (Recall,
  in particular, that if $x^i$ partially edit matches to a block, then
  $x^i$ must be within $pB/8$ edits of a string containing that
  block.) Thus $w_t$ must each be within edit distance $pB / 8 + pB /
  100 \le pB / 4$ of a substring $b \in \{b_1, b_2\}$ of $x^i$. Since
  $\min(\ed(w_t, \hat{y}^k), \ed(w_t, \hat{y}^l)) \le pB / 100$ and
  $\min(\ed(b, \hat{y}^k), \ed(b, \hat{y}^l)) \le pB / 8$, it must be
  that $w_t$ is of length within $pB / 100$ of $B$ and $b$ is of
  length within $pB/8$ of $B$. It follows that there is a $B$-letter
  substring of $x^i$ within edit distance $pB / 100 + pB / 8 +
  pB / 4 < pB$ of a $B$-letter substring not intersecting $x^i$. This
  prevents $x^i$ from being $p$-unique, a contradiction.
\end{proof}

Armed with the preceding lemmas, the next lemma presents the basic mechanism through which a solution
to the Clean Alignment Problem can be transformed into a
solution to the Pseudorandom Edit Distance Problem. (Note that the
lemma does not use the fact that the output of the Clean
Alignment Problem is an alignment rather than a matching; in fact it
turns out the non-crossing property of the output is not necessary for
our reduction.)

\begin{lem}
Let $E$ be a set of disjoint edges between $\{x^i\}$ and $\{y^i\}$
such that for all $e = (x^i, y^j) \in E$, $x^i$ edit matches with
$y^j$. Let $t$ be the number of letters in $x$ which reside in a block
unmatched by $E$. Then in time $O(Bn \log n)$, one can recover from $E$ a
sequence of $O\left(\frac{1}{p}\ed(x, y) + M + t\right)$ edits from
$x$ to $y$.
\label{lemfinalrecovery}
\end{lem}
\begin{proof}
Because $\sum_i c_i = \ed(x, y)$, all but $O\left(\frac{1}{p} \ed(x, y) \right)$ of the letters in $x$ are
contained in $x^i$'s satisfying $c_i \le pB / 100$. Moreover, by Lemma
\ref{lempolygamous}, at most $O\left(\frac{1}{p}\ed(x, y) + M\right)$
letters are contained in polygamous $x^i$'s. 

Now consider a monogamous $x^i$ for which $c_i \le pB / 100$. We claim
that $x^i$ must partially edit match to a $y^j$ which is within fewer
than $6B$ positions of where $x^i$ is mapped to by $\mathcal{A}$
(meaning each letter in $x^i$ that is matched by $\mathcal{A}$ is
matched to a letter within fewer than $6B$ positions of $y^j$). In
particular, if $\mathcal{A}$ transforms $x^i$ into a substring of $y$
beginning at some $y_k$, then because $c_i \le pB / 100$, it must be
that $\ed(x^i, y_k \cdots y_{k + 6B - 1}) \le pB / 50$. If we define
$k'$ to be $k$ rounded down to the nearest position whose zero-indexed
position is a multiple of $pB / 100$, then it follows that $\ed(x^i,
y_{k'} \cdots y_{k' + 6B - 1}) \le pB / 25$. If we define $k''$ to be
$k'$ rounded up to the next multiple of $3B$, then $x^i$ partially
edit matches with $y^{k'' / 3B}$. By construction of $y^{k'' / 3}$,
each letter in $x^i$ that is matched by $\mathcal{A}$ is matched to a
letter within fewer than $6B$ positions of $y^{k'' / 3B}$.

Define $j$ to be $k'' / 3B$. Since $x^i$ partially edit matches with
$y^j$ and $x^i$ is monogamous, the unique $y^{j'}$ to which $x^i$ edit
matches must satisfy $j' \in \{j - 1, j\}$. This, in turn, means that
every letter in $x^i$ which is matched by $\mathcal{A}$ to a letter in
$y$ is matched to a letter in $y$ within fewer than $9B$ positions of
$y^{j'}$.

Recall that all but at most $O\left(\frac{1}{p}\ed(x, y) + M\right)$
of the letters in $x$ reside in blocks $x^i$ which are monogamous and
satisfy $c_i \le pB / 100$. Moreover, of these letters, all but at
most $t$ of them reside in a block $x^i$ which is edit matched by
$E$. Hence, with the exception of $x^i$'s containing a total of at
most $O(\frac{1}{p}\ed(x, y) + M + t)$ letters, every $x^i$ is edit
matched by $E$ to a single $y^j$, and any edges in the alignment
$\mathcal{A}$ between a letter in $x^i$ and a letter in $y$ have to
concern a letter in $y$ within fewer than $9B$ positions of $y^j$.

Consider the set $T$ of edges between letters in $x$ and $y$, where an
edge $(x_r, y_s)$ is included if $x_r$ is in some $x^i$ which is
matched by $E$ to some $y^j$ within fewer than $9B$ positions of
$y_s$. Note that $T$ contains at most $O(nB)$ edges. We have shown
that the intersection between $T$ and the edges in the alignment
$\mathcal{A}$ yields an alignment which matches all but at most
$O(\ed(x, y) + M + t)$ letters of $x$ (and thus also of $y$ since $|y|
= |x|$). Applying Lemma 3.1 of \cite{smoothcomplexity}, which finds
the optimal alignment restricted to a set $T$ of edges in time $O(|T|
\log n)$, it follows that we can recover a sequence of $O(\ed(x, y) +
M + t)$ edits from $x$ to $y$ in time $O(Bn \log n)$.
\end{proof}

To complete the proof of Theorem \ref{thmreduction} using Lemma
\ref{lemfinalrecovery}, it suffices to show that any optimal clean
edit-matching alignment between $(x^1, \ldots, x^{n/ 6B})$ and $(y^1,
\ldots, y^{n / 3B})$ leaves no more than $O(\ed(x, y) + M)$ letters
unmatched.

We begin by analyzing the cost of an optimal (possibly crossing)
edit matching. An \emph{edit matching} between $(x^1, \ldots, x^{n/
  6B})$ and $(y^1, \ldots, y^{n / 3B})$ is a matching whose edges are
between edit-matched $x^i$'s and $y^j$'s. The \emph{cost} of an edit
matching is defined in the same asymmetric manner as the cost of an edit-matching
alignment.

\begin{lem}
  Let $T$ be a minimum-cost edit matching between $(x^1, \ldots, x^{n/
    6B})$ and $(y^1, \ldots, y^{n / 3B})$. Then the cost of $T$ is at
  most $$O\left( \frac{\frac{1}{p}\ed(x, y) + M}{B}\right).$$
  \label{lemmincostmatching}
\end{lem}
\begin{proof}
  Recall that any block $x^i$ with $c_i \le pB / 100$ partially edit
  matches to some $y^j$, and thus also edit matches to some
  $y^j$. (See the proof of Lemma \ref{lemfinalrecovery}.) It follows
  that all but $O\left( \frac{1}{pB} \ed(x, y) \right)$ of the blocks
  $x^i$ edit match with some $y^j$.

  By Lemma \ref{lemnotpunique}, all but $O\left(\frac{1}{pB} \ed(x, y)
  + M / B\right)$ of the $x^i$'s edit match with some $y^j$ that
  doesn't edit match with any other $x^{i'}$'s. It follows that there
  is an edit matching between $(x^1, \ldots, x^{n/ 6B})$ and $(y^1,
  \ldots, y^{n / 3B})$ in which all but $O\left(\frac{1}{pB} \ed(x, y)
  + M / B\right)$ of the $x^i$'s are matched.

  The $x$-portion cost of such an edit matching $T$ is at most
  $O\left(\frac{1}{pB} \ed(x, y) + M / B\right)$.  Lemma
  \ref{lemfinalrecovery} can be applied to $T$ to recover a sequence
  of $O\left(\frac{1}{p}\ed(x, y) + M\right)$ edits from $x$ to
  $y$. Recall that the sequence of edits recovered by Lemma
  \ref{lemfinalrecovery} removes any letters in $y$ which are not
  within fewer than $9B$ positions of some $y^i$ that is matched by
  $T$. Hence, given any seven adjacent $y^i$'s left unmatched by $T$ (corresponding with $21B$ unmatched letters),
  the middle $y^i$ of the seven must incur $3B$ edits in the process described by Lemma \ref{lemfinalrecovery}. Since only
  $O\left(\frac{1}{p}\ed(x, y) + M\right)$ edits can be incurred in
  total, the $y$-portion cost of $T$ cannot exceed
  $$O\left(\frac{\frac{1}{p}\ed(x, y) + M}{B}\right).$$
  This completes the proof of the lemma.
\end{proof}

The preceding lemma bounds the cost of a minimum-cost edit
matching. The next lemma will allow us to bound the cost of a
minimum-cost edit-matching alignment. In particular, the lemma shows
that in any edit matching, the crossings can be eliminated by removing
only edges that are incident to an $x^i$ with one of three properties:
$x^i$ is either polygamous, non-$p$-unique, or satisfies $c_i \ge pB /
100$. Since one can bound the number of such $x^i$'s, it follows that
we can transform any edit matching into an edit-matching alignment by
removing only a small number of edges.

\begin{lem}
  Suppose that for $i < j$ and $k \le l$, $x^i$ edit matches with
  $x^l$ and $x^j$ edit matches with $x^k$. Then either $\max(c_i ,c_j)
  \ge pB / 100$, or at least one of $x^i$ or
  $x^j$ are either polygamous or non-$p$-unique.
  \label{lemremovecrossings}
\end{lem}
\begin{proof}
  If both $c_i$ and $c_j$ are smaller than $pB / 100$, then the
  alignment $\mathcal{A}$ tells us how to partially edit match $x^i$
  and $x^j$ each to some $y^q$ and $y^r$ with $q \le r$ (as done in
  the proof of Lemma \ref{lemfinalrecovery}). If $q = r$,then by Lemma
  \ref{lemnotpunique}, neither $x^i$ nor $x^j$ are $p$-unique. If $q <
  r$ and both $x^i$ and $x^j$ are monogamous, then since $k \le l$, it
  follows that $k = l$. But by Lemma \ref{lemnotpunique}, this
  prevents either $x^i$ or $x^j$ from being $p$-unique.

  Therefore, if both $c_i$ and $c_j$ are less than $pB / 100$, then at
  least one of $x^i$ or $x^j$ are either polygamous or non-$p$-unique.
\end{proof}



Finally, we can now bound the cost of an optimal clean edit-matching alignment.
\begin{lem}
Let $T$ be a minimum-cost clean edit-matching alignment between $(x^1,
\ldots, x^{n/ 6B})$ and $(y^1, \ldots, y^{n / 3B})$. Then the cost of
$T$ is at most $$O\left( \frac{\frac{1}{p}\ed(x, y) + M}{B}\right).$$
\label{lemcostcleanalignment}
\end{lem}
\begin{proof}
  Consider an optimal edit-matching $T_0$. By Lemma
  \ref{lemmincostmatching}, the cost of $T_0$ is no greater than $O\left(
  \frac{\frac{1}{p}\ed(x, y) + M}{B}\right)$.

  To obtain $T$, remove from $T_0$ any edges involving some $x^i$ such
  that either $c_i < pB / 100$, $x^i$ is polygamous, or $x^i$ is
  non-$p$-unique. This removes at most $O\left(\frac{\ed(x,
    y)}{pB}\right)$ $x^i$'s with $c_i < pB / 100$, at most
  $O\left(\frac{1}{pB}\ed(x, y) + M / B\right)$ $x^i$'s which are
  polygamous (Lemma \ref{lempolygamous}), and at most $M/B$ $x^i$'s
  which are not $p$-unique. Since each edge-removal can increase the
  $x$-portion and $y$-portion costs each by at most one, the cost of
  $T$ is at most $O\left( \frac{\frac{1}{p}\ed(x, y) + M}{B}\right)$.

  By Lemma \ref{lemremovecrossings}, $T$ has no crossings. Moreover,
  because $T$ does not match any polygamous $x^i$'s, and because by
  Lemma \ref{lemnotpunique} $T$ does not match any $y^j$'s that edit
  match to more than one $x^i$, $T$ is a clean edit-matching
  alignment.
\end{proof}

Theorem \ref{thmreduction} follows easily from the preceding lemmas.
\begin{proof}[Proof of Theorem \ref{thmreduction}]
  We will assume without loss of generality that $p \ge 1/B$, since
  otherwise we may replace $p$ with $1/B$ without changing the
  definition of $p$-unique blocks.
  
  Suppose there is an algorithm for the Clean Alignment Problem that
  runs in time $\tilde{O}(n)$ and returns a clean edit-matching
  alignment whose cost is within a factor of $O(1)$ of optimal, in
  expectation.

  For a given $i$ and $j$, determining whether $x^i$ edit matches $y^j$ can be done in time
  $O\left(B^2\right)$. In particular, we need only compute the edit distances between $x^i$ and $O(1 / p)$ substrings of $y$. Denote these substrings by $a_1, \ldots, a_{q}$. For each $a_s$, we wish to determine whether $\ed(x^i, a_s) \le pB / 8$. Using the exact algorithm of \cite{landau1998incremental}, which computes $\ed(u, v)$ in time $O(|u| + |v| + \ed(u, v)^2)$, determining whether $\ed(x^i, a_s) \le pB / 8$ can be done in time $O\left( B + \left(pB\right)^2\right) \le O(pB^2)$. Performing $O(1/p)$ such computations in order to determine whether $x^i$ edit matches with $y^j$ therefore takes time $O(B^2)$.
  
 Thus the algorithm for the Clean Alignment Problem can be employed
  to find a clean edit-matching alignment between $(x^1, \ldots, x^{n/ 6B})$ and $(y^1, \ldots, y^{n
    / 3B})$ in time 
    $$\tilde{O}\left( (n / B) B^2\right)
  = \tilde{O}\left(nB \right).$$ The cost $t$ of this alignment will have expected value at most
  $$O\left( \frac{\frac{1}{p}\ed(x, y) + M}{B}\right),$$
  by Lemma \ref{lemcostcleanalignment}.

  By Lemma \ref{lemfinalrecovery}, one can then recover in time $O(Bn
  \log n)$ a sequence of $t'$ edits from $x$ to $y$ satisfying
  $$\E[t'] \le O\left(\frac{1}{p}\ed(x, y) + M\right),$$
  as desired.  
\end{proof}

\begin{remark}
  When $p^2 \ge \frac{1}{B}$, we can use the bound $O(B + (pB^2)) \le
  O(p^2B^2)$ (rather than bounding the quantity by $pB^2$) in order to
  prove a slightly better runtime of $\tilde{O}(p \cdot nB)$.
\end{remark}

\section{Determining parameters $p$ and $B$}\label{secoblivious}

In this section, we consider the situation where the parameters $p$
and $B$ are not known beforehand. Throughout the section, for a string
$x$, and parameters $p$ and $B$, we define $M_{p, B}(x)$ (often
abbreviated by $M_{p, B}$) to be the quantity $M(x)$ defined in
Section \ref{secmain} (in which the parameters $p$ and $B$ were
implicit).

We consider two settings. In the first setting, we are given two
strings $x$ and $y$ and wish to compute $\ed(x, y)$ within a factor of
$O(\alpha)$ as fast as possible. Theorem \ref{thmsingleshot} gives an
algorithm for this problem which runs in time
$\tilde{O}(n^{4/3}B^{2/3})$, where $B$ is the
smallest block size for which the string $x$ is $(\alpha,
B)$-pseudorandom.

\begin{thm}
Let $x$ and $y$ be strings in $\Sigma^n$. Let $\alpha > 1$ be an
approximation threshold, and suppose there is a power-of-two
block-size $B$ for which the string $x$ is $(1/\alpha,
B)$-pseudorandom. Then without knowing $B$ beforehand, one can still
(with high probability) approximate $\ed(x, y)$ within a factor of
$O(\alpha)$ in time $\tilde{O}(n^{4/3}B^{1/3})$.

Moreover, rather than $x$ being $(1/ \alpha, B)$-pseudorandom, it
suffices that $M_{1/\alpha, B}(x) < n^{2/3} \cdot B^{1/3}$.
\label{thmsingleshot}
\end{thm}

In the second setting we consider, one is given a source string $x$
which one intends to compare with many other strings $y$. Defining
$\alpha$ and $B$ as before, we present an $\tilde{O}(n^{1.5}(\sqrt{B} +
\alpha))$-time algorithm for preprocessing $x$, so that any edit
distance computation $\ed(x, y)$ can later be $\alpha$-approximated in
time $\tilde{O}(nB)$.

\begin{thm}
Let $x$ be a string in $\Sigma^n$. Let $\alpha > 1$ be an
approximation threshold, and suppose there is a power-of-two
block-size $B$ for which the string $x$ is $(1/\alpha,
B)$-pseudorandom. Then without knowing $B$ beforehand, one can (with high
probability) construct an algorithm $\mathcal{S}$ in time
$\tilde{O}(n^{1.5}(\sqrt{B} + \alpha))$ so that $\mathcal{S}$ is an
$\tilde{O}(nB)$-time algorithm for approximating $\ed(x, y)$ within a
factor of $O(\alpha)$ (with high probability) for arbitrary strings $y
\in \Sigma^n$.

Moreover, rather than $x$ being $(1/ \alpha, B)$-pseudorandom, it
suffices that $M_{1/\alpha, B}(x) < n^{1/2} \cdot B^{1/2}$.
\label{thmmultishot}
\end{thm}

The proofs of Theorems \ref{thmsingleshot} and \ref{thmmultishot} will
rely on the following lemma:

\begin{lem}
    \label{lemtestp}
  Consider a string $x$ of length $n$, and a parameter $B \in
  \mathbb{N}$ such that $B$ is a power of two. Then for any threshold
  $n^\epsilon$ and value $p = \frac{1}{2^j}$, there is an algorithm
  which runs in time $\tilde{O}(n^{2 - \epsilon} (B + 1/p^2))$, and
  which with high probability returns true if $M_{p / 2, B} \le
  n^{\epsilon} / 2$, and false if $M_{p, B} >
  2n^{\epsilon}$. (Otherwise, the algorithm may return true or false
  arbitrarily.)
\end{lem}
\begin{proof}
  We will assume without loss of generality that $n$ is a multiple of
  $6B$, since otherwise $x$ can be padded with null characters to be a
  length divisible $6B$, without changing $M_{p, B}$. Let $x^1 \cdots
  x^{n / 6B}$ be the $6B$-letter blocks that make up $x$.
  
  Let $c$ be a sufficiently large constant. The algorithm uses a
  simple random sampling approach. We randomly sample $s = c n^{1 -
    \epsilon} \log n$ blocks $x^{i_1}, \ldots, x^{i_s}$ from
  $x$. For each of these blocks $x^{i_l}$, we will perform a test on
  $x^{i_l}$ which returns true if $x^{i_l}$ is $(p/2)$-unique, returns
  false if $x^{i_l}$ is not $p$-unique, and may return either true or
  false otherwise. (We will describe how this test, which we call the
  \emph{uniqueness test}, works at the end of the proof.) Let $r$ be
  the number of $x^{i_l}$'s for which the uniqueness test returns
  false. Because the uniqueness test returns true whenever $x^{i_l}$
  is $(p / 2)$-unique,
  $$\E[r] \le s \cdot \frac{M_{p / 2, B}}{n} = \frac{M_{p / 2, B}
    \cdot c \log n}{n^{\epsilon}}.$$ Thus if $M_{p / 2, B} \le
  n^\epsilon / 2$, then $\E[r] \le \frac{c}{2} \log n$, which by a
  Chernoff bound means that (for $c$ large enough) with high
  probability $r < c \log n$. On the other hand, because the
  uniqueness test returns false whenever $x^{i_l}$ is not $p$-unique,
  $$\E[r] \ge s \cdot \frac{M_{p, B}}{n} = \frac{M_{p, B} \cdot c \log
    n}{n^{\epsilon}}.$$ Thus if $M_{p, B} \ge 2n^{\epsilon}$, then
  $\E[r] \ge 2c \log n$, which by a Chernoff bound means that (for $c$
  large enough) $r > c \log n$ with high probability. It follows that
  by returning whether $r \le c \log n$, our algorithm will with high
  probability return true if $M_{p / 2, B} \le n^{\epsilon} / 2$, and
  return false if $M_{p, B} > 2n^{\epsilon}$.

  In the remainder of the proof, we describe the uniqueness test for a
  given block $x^{i_l}$, and evaluate the runtime of the
  algorithm. The uniqueness test for $x^{i_l}$ examines every
  $B$-letter substring $a$ in $x$ such that $a$ does not intersect
  $x^{i_l}$ and such that $a$ ends in a position divisible by
  $\frac{pB}{8}$; and also examines each $B$-letter substring $b$ in
  $x$ that is contained entirely in $x^{i_l}$ and ends in a position
  divisible by $\frac{pB}{8}$. For each such pair $a$ and $b$, the
  algorithm uses the low-distance regime algorithm of
  \cite{landau1998incremental} to determine in time $O(B + p^2B^2)$
  whether $\ed(a, b) \le pB / 2$. The uniqueness test returns true if
  and only if $\ed(a, b) \le p B/ 2$ for all such $a$ and $b$. Notice
  that if $x^l$ is $(p/2)$-unique, then the uniqueness test will
  necessarily return true. Moreover, if the uniqueness test returns
  true, then for every $B$-letter substring $a'$ in $x$ not
  intersecting $x^l$, and every $B$-letter substring $b'$ in $x^l$, it
  must be that $\ed(a, b) \le pB / 2 + 4 \frac{pB}{8} \le pB$, meaning
  that $x^l$ is $p$-unique. Hence the uniqueness test satisfies both
  properties required from it.

  Finally, we evaluate the runtime of our algorithm. Since there are
  $O(n / (pB))$ options for $a$ and $O(1 / p)$ options for $b$ in each
  uniqueness test, each uniqueness test runs in time
  $$O\left(\frac{n}{pB} \cdot \frac{1}{p} \cdot (B + p^2B^2)\right) \le O(nB + n /
  p^2).$$ Since our algorithm runs $\tilde{O}(n^{1 - \epsilon})$
  uniqueness tests, the full runtime is
  $\tilde{O}(n^{2 - \epsilon} (B + 1/p^2))$.
\end{proof}

Using Lemma \ref{lemtestp}, we now prove Theorem \ref{thmsingleshot}.

\begin{proof}[Proof of Theorem \ref{thmsingleshot}]
  Assume without loss of generality that $\alpha$ is a power of two.
  To accomplish the desired time bound, we run two algorithms in
  parallel. The first algorithm is the low-distance regime algorithm
  of \cite{landau1998incremental} which runs in time $O(n + \ed(x,
  y)^2)$. The second algorithm considers values $B_i = \alpha \cdot
  2^i$ for $i = 0, 1, 2, \ldots$, and applies Lemma \ref{lemtestp} to
  each $B_i$ with $p = \frac{2}{\alpha}$ and $n^{\epsilon} = 2n^{2/3}
  \cdot \alpha^{2/3} \cdot (B_j + \alpha^2)^{1/3}$; upon finding a
  $B_j$ for which Lemma \ref{lemtestp} returns true, the algorithm
  approximates $\ed(x, y)$ in time $\tilde{O}(n \cdot B_j)$ by
  treating $x$ as $(2/\alpha, B_j)$-pseudorandom.

  If the low-distance regime algorithm finishes first, then we can
  return $\ed(x, y)$ exactly. If the second algorithm terminates
  first, then we continue to run the low-distance regime algorithm
  until it has run for time $\Theta(n^{4/3} \cdot \alpha^2 \cdot (B_j
  + \alpha^2)^{2/3})$, and only if the low-distance regime algorithm
  does not terminate in that time window do we use the output of the
  second algorithm.

  Because $p = \frac{2}{\alpha}$ and $n^{\epsilon} = 2n^{2/3} \cdot
  \alpha^{2/3} \cdot (B_j + \alpha^2)^{1/3}$, the second algorithm's
  runtime will be
  \begin{equation}
    \begin{split}
      & \tilde{O}(n^{2 - \epsilon} (B_j + 1/p^2)) \\
      & = \tilde{O}(n^{4/3} (B_j + \alpha^2)^{2/3} / \alpha^{2/3})  \\
      &\le
      \tilde{O}(n^{4/3} (B_j^{2/3} + \alpha^{2/3})) \\
            &\le
    \tilde{O}(n^{4/3} B_j^{2/3}).
    \end{split}
    \label{eqbackhurts}
  \end{equation}
  Notice that Lemma \ref{lemtestp} guarantees with high probability
  that $B_j \le B$ (because $M_{1/\alpha, B} < n^{2/3} \cdot
  B^{1/3}$), meaning that with high probability our algorithm achieves
  runtime $\tilde{O}(n^{4/3}B^{2/3})$, as desired.

  To prove the accuracy of our algorithm, consider the case where the
  output is determined by the second algorithm. By Lemma
  \ref{lemtestp}, with high probability, $M_{2/\alpha, B_j} \le
  2n^\epsilon = 4n^{2/3} \cdot \alpha^{2/3} \cdot (B_j +
  \alpha^2)^{1/3}$. By (the high-probability version of) Theorem
  \ref{thmmain}, it follows that the second algorithm will find
  a sequence of $t$ edits for some $t$ satisfying
  \begin{equation}
    t \le O(\alpha \ed(x, y) + n^{2/3} \cdot \alpha^{2/3} \cdot (B_j +
    \alpha^2)^{1/3}),
    \label{eqbackhurts2}
  \end{equation}
  with high probability. Moreover, because the low-distance regime
  algorithm did not finish in the time $O(n^{4/3} (B_j +
  \alpha^2)^{2/3} / \alpha^{2/3})$, it must be that $\ed(x, y) \ge
  n^{2/3} \cdot (B_j + \alpha^2)^{1/3} / \alpha^{1/3}$. Plugging this into
  Equation \eqref{eqbackhurts2}, we get that $t \le O(\alpha \ed(x,
  y))$, as desired.
\end{proof}

Finally, we prove Theorem \ref{thmmultishot}.

\begin{proof}[Proof of Theorem \ref{thmmultishot}]
  We assume without loss of generality that $\alpha$ is a power of
  two. Our preprocessing step goes as follows. We consider values $B_i
  = \alpha 2^i$ for $i = 0, 1, 2, \ldots$, and apply Lemma
  \ref{lemtestp} to each $B_i$ with $p = \frac{2}{\alpha}$ and
  $n^{\epsilon} = 2n^{1/2} \cdot (B_j + \alpha^2)^{1/2}$; upon finding
  a $B_j$ for which Lemma \ref{lemtestp} returns true, we have
  completed the preprocessing step. Note that the preprocessing step
  takes time
  \begin{equation*}
    \begin{split}
      & \tilde{O}(n^{2 - \epsilon}(B_j + \alpha^2)) \\
      & = \tilde{O}(n^{1.5}\sqrt{B_j + \alpha^2}), \\
    \end{split}
  \end{equation*}
 and since with high probability $B_j$ will
  be at most $B$ (by Lemma \ref{lemtestp} and because of the fact that
  $M_{1/\alpha, B} < n^{1/2} \cdot B^{1/2}$), it follows that with
  high probability the runtime of the preprocessing step is at most
  $\tilde{O}(n^{1.5}\sqrt{B + \alpha^2})$, as desired.

  Given a string $y$, we can then approximate $\ed(x, y)$ as
  follows. We first run the low-distance regime algorithm of
  \cite{landau1998incremental} for time $\Theta(n \cdot B_j)$, and if
  it completes then we return its output. Otherwise, we may conclude
  that
  \begin{equation}
    \ed(x, y) \ge \sqrt{n B_j}.
    \label{eqbackhurtsless}
  \end{equation}
    In the latter case, we treat $x$ as $(2/\alpha, B_j)$-pseudorandom
    and generate an estimate for $\ed(x, y)$ in time $\tilde{O}(n
    \cdot B_j)$ (using the high-probability version of Theorem
    \ref{thmmain}). By Lemma \ref{lemtestp}, we know that with high
    probability $M_{2/\alpha, B_j} \le 2n^\epsilon = 4 n^{1/2} \cdot
    (B_j + \alpha^2)^{1/2}$. By
    Theorem \ref{thmmain}, it follows that with high probability the
    number of edits $t$ which we find between the strings $x$ and $y$
    will satisfy
    \begin{align*}
      t & \le O\left(\alpha \ed(x, y) + n^{1/2} \cdot (B_j + \alpha^2)^{1/2}\right) \\
      & \le O\left(\alpha \ed(x, y) + \alpha \cdot n^{1/2} \cdot B_j^{1/2}\right).
    \end{align*}
    By \eqref{eqbackhurtsless}, we get $t \le
    O(\alpha \ed(x, y))$, as desired.
\end{proof}

\section{Conclusion}

In this paper, we have proposed a natural model of pseudorandomness
for strings which allows for edit distance to be efficiently
approximated. We conclude by presenting two directions for future
work.
\begin{itemize}
  \item One important direction for future work is to investigate the
    degree to which the algorithm presented is practical for
    real-world inputs. In particular, for what parameters $p$ and $B$
    do real-world inputs tend to be (almost) $(p, B)$-pseudorandom?

  \item On the theoretical side, it is interesting to note that in the
smoothed complexity model the algorithm of \cite{smoothcomplexity}
achieves sublinear runtime for strings $x$ and $y$ of sufficiently
large edit distance apart. Can such a result be replicated in our
pseudorandomness model?
\end{itemize}

\paragraph*{Acknowledgments} The author would like to thank Moses Charikar for his mentoring and advice throughout the project. The author would also like to thank Ofir Geri for many useful conversations, and Michael A. Bender for his advice during the writing process.

This research was supported in part by NSF Grants 1314547 and 1533644.

\bibliographystyle{plain} \bibliography{ref}

\end{document}